\documentclass[11pt]{article}

\usepackage{tikz}
\usetikzlibrary{shapes.geometric}
\usepackage{amsmath}
\usepackage{amsthm}
\usepackage{amsfonts}
\usepackage{amssymb}
\usepackage{array}
\usepackage[utf8]{inputenc}
\usepackage{geometry}
\usepackage{graphicx}
\usepackage{algorithm}
\usepackage{algorithmic}
\usepackage[T1]{fontenc}
\usepackage[utf8]{inputenc}
\usepackage{authblk}
\usepackage[english]{babel}
\newtheorem{theorem}{Theorem}[section]

\newtheorem{proposition}[theorem]{Proposition}
\newtheorem{corollary}[theorem]{Corollary}
\newtheorem{lemma}[theorem]{Lemma}

\theoremstyle{definition}
\newtheorem{definition}[theorem]{Definition}
\newtheorem{example}[theorem]{Example}
\newtheorem{remark}[theorem]{Remark}
\newtheorem{construction}{Construction}
\newtheorem*{conjecture}{Conjecture}

\newcommand{\C}{\mathcal{C}}

\newcommand{\F}{\mathbb{F}}

\newcommand{\wt}{\mathrm{wt}}

\newcommand{\M}{\mathcal M}
\newcommand{\Fq}{\mathbb{F}_q}

\newcommand{\ch}{\mathrm{Char}}

\newcommand{\supp}{\mathrm{supp}}

\title{A geometric characterization of minimal codes and their asymptotic performance}

\author[1]{Gianira N. Alfarano \thanks {G. N. Alfarano acknowledges the support of  Swiss National Science Foundation grant n. 188430.}}

\affil[1]{University of Z{u}rich, Switzerland}

\author[2]{Martino Borello}

\affil[2]{LAGA,  UMR 7539, CNRS, Universit\'e Paris 13 - Sorbonne Paris Cit\'e, Universit\'e Paris 8, F-93526, Saint-Denis, France}

\author[3]{Alessandro Neri \thanks {A. Neri acknowledges the support of  Swiss National Science Foundation grant n. 187711.}}

\affil[3]{Inria Saclay \^Ile-de-France, 91120 Palaiseau, France}

\begin{document}
\maketitle

\begin{abstract}
In this paper, we give a geometric characterization of minimal linear codes. In particular, we relate minimal linear codes to cutting blocking sets, introduced in a recent paper by Bonini and Borello.
Using this characterization, we derive some bounds on the length and the distance of minimal codes, according to their dimension and the underlying field size. Furthermore, we show that the family of minimal codes is asymptotically good. Finally, we provide some geometrical constructions of minimal codes as cutting blocking sets.

\end{abstract}

\section{Introduction}


Let $\Fq$ be a finite field and $\C\subseteq \Fq^n$ be a linear code. A codeword $c\in\C$ is called \emph{minimal} if its support $\{i \mid c_i\neq 0\}$
does not contain the support of another independent codeword. The study of the minimal codewords of a linear code finds application in combinatorics, in the analysis of the Voronoi region for decoding purposes \cite{ashikhmin1998, agrell1998voronoi} and in secret sharing schemes \cite{massey1993minimal, massey1995some, ashikhmin1998}. 

Secret sharing schemes were introduced independently by Shamir and Blakley in 1979 \cite{shamir1979share, blakley1979safeguarding}. They are protocols used for distributing a secret among a certain number of participants. In particular, in its original 
framework, a secret sharing scheme works as follows: a dealer gives a share of a secret to $n$ players in such a way that any subset of at least  $t$ players can reconstruct the secret, but no subset of less than $t$ players can. This is also called $(n,t)$-threshold scheme protocol. A more general construction, based on linear codes, was first investigated by McEliece and Sarwate in 1981 \cite{mceliece1981sharing}, where Reed-Solomon codes were used. Later, several authors used other linear error-correcting codes to construct the same protocol \cite{karnin1983secret, massey1993minimal, massey1995some, ding2000secret}. 

The set of subsets of participants which are able to recover the secret is called \emph{access structure}. It is common to consider only subsets which do not admit proper subsets of participants able to recover the secret: we may refer to their collection as \emph{minimal access structure}. For example, in an $(n,t)$-threshold scheme protocol, the access structure is given by all subsets of at least $t$ participants, whereas the minimal access structure is given by all subsets of exactly $t$ participants.
 
In \cite{massey1993minimal}, Massey relates the secret sharing protocol to minimal codewords: in particular, the minimal access structure in his secret sharing protocol is given by the support of the minimal codewords of a linear code $\C$, having first coordinate equal to $1$. However, finding the minimal codewords of a general linear code is a difficult task. For this reason, a special class of codes has been introduced: a linear code is said to be \emph{minimal} if all its nonzero codewords are minimal.


In \cite{ashikhmin1998}, Ashikhmin and Barg gave a sufficient condition for a linear code to be minimal.

\begin{lemma}
Let $\C$ be an $[n,k]_q$ code, $w_{min}, w_{max}$ be the minimum and the maximum Hamming weights in $\C$, respectively. Then $\C$ is minimal if 
\begin{equation}\tag{AB}\label{AB}
\frac{w_{min}}{w_{max}}> \frac{q-1}{q}. 
\end{equation}
\end{lemma}

The Ashikhmin-Barg Lemma gave rise to several works with the aim of constructing minimal codes, see for example \cite{carlet2005linear, yuan2005secret, ding2015linear, ding2016three}. However, condition \eqref{AB} is only sufficient. Some constructions of families of minimal codes not satisfying the condition \eqref{AB} were first presented in \cite{cohen2013minimal,chang2018linear}. In \cite{heng2018minimal}, a necessary and sufficient condition for an $\Fq$-linear code to be minimal was given:
an $[n,k]_q$ code $\C$ is minimal if and only if, for every pair of linearly independent codewords $a,b\in \C$, we have
$$\sum_{\lambda \in \Fq^*} \wt(a+\lambda b) \ne (q-1)\wt(a)-\wt(b).$$

In the same paper, the authors constructed an infinite family of minimal linear codes not satisfying the condition \eqref{AB}. This construction was generalized to finite fields with odd characteristic by Bartoli and Bonini, in \cite{bartoli2019minimalLin}. In \cite{bonini2019minimal}, Bonini and Borello investigated the geometric generalization of the construction in \cite{bartoli2019minimalLin}, highlighting a first link between minimal codes and cutting blocking sets. Moreover, different types of recent constructions of minimal codes based on weakly regular bent plateaued functions have been also presented in \cite{sihem1,sihem2,sihem3}.

In this paper, we give a characterization of minimal linear codes in terms of cutting blocking sets.  We derive some bounds on the length and the distance of minimal codes, according to their dimension and the underlying field size. We then show that the family of minimal codes is asymptotically good and we provide some geometrical constructions and examples.
The paper is organized as follows. In Section \ref{sec:Preliminaries}, we introduce some basics about linear codes over finite fields. In particular, we focus on minimal codes and we introduce the notion of reduced minimal codes. After giving some background on projective systems, we explain how they are in one-to-one correspondence with linear codes. In Section \ref{sec:cuttingsets}, we relate minimal codes (reduced minimal code resp.) to cutting blocking sets (minimal cutting blocking sets resp.), by analyzing the correspondence given in Section \ref{sec:Preliminaries}. In Section \ref{sec:bounds}, we derive bounds on the distance and on the length of minimal codes. One of the main results of this section is Theorem \ref{thm:asymgood}, in which we show that minimal codes are asymptotically good. Moreover, we find a correspondence between cutting blocking sets in $PG(2,q)$ and 2-fold blocking sets and we use this correspondence to derive upper and lower bounds on the length of reduced minimal codes of dimension $3$. In Section \ref{sec:constructions}, we provide a geometrical general construction of reduced minimal codes and we compute the weight distribution of these codes. For reduced minimal codes of dimension $4$ and minimal codes of dimension $5$, we exhibit a construction exploiting cutting blocking sets in $PG(3,q)$ and in $PG(4,q)$ (with smaller length than the ones derived with the general construction). We conclude with further remarks and open questions in Section \ref{sec:conclusion}, where we also propose a challenging conjecture on the minimum distance of minimal codes.

\section{Preliminaries}\label{sec:Preliminaries}

\subsection{Linear codes}

We recall here some basic notions in coding theory which will be useful in the following. 

Let $q$ be a prime power, $n$ be a positive integer and $\Fq$ be the finite field with $q$ elements. In the vector space $\Fq^n$, the \emph{support} of a vector $u=(u_1,\ldots, u_n) \in \Fq^n$ is the set $\supp(u):=\{ i \mid u_i \neq 0\}$. The \emph{Hamming distance} on $\Fq^n$ is defined as $d_H(u,v)=|\supp(u-v)|$, for every pair of vectors $u,v \in \Fq^n$. The \emph{Hamming weight}, $w(u)$ of a vector $u \in \Fq^n$ is its distance from the all zero vector.

 An $[n,k]_q$ \emph{(linear) code} $\C$  is a $k$-dimensional subspace of $\Fq^n$ endowed with the Hamming distance and the elements of $\C$ are  called \emph{codewords}. Its \emph{rate} is the number $R=k/n$. The \emph{minimum distance} $d$ of $\C$ is the quantity $d=\min\{d_H(u,v)\mid u,v \in \C, u\neq v\}$. If the minimum distance $d$ of an $[n,k]_q$ code $\C$  is known, then $\C$ is denoted as $[n,k,d]_q$ code. The weight distribution of $\C$ is the sequence $A_0(\C),\ldots,A_n(\C)$, where $A_i(\C)=|\{c\in C\mid w(c)=i\}|$.

An $[n,k]_q$ code  $ \C $ of dimension $ k\ge 2 $ is said to be   \emph{non-degenerate} if no coordinate position is identically zero. Unless specified otherwise, all codes discussed here are assumed to be non-degenerate.

An important notion for linear codes concerns the equivalence.
%
 Let $\mathcal G$ be the subgroup of the group of linear automorphisms of $\F_q^n$ generated by the permutations of coordinates and by the multiplication of the $i$-th coordinate by an element in $\F_q^*$. Two codes $\C$ and $\C'$ are  \emph{(monomially) equivalent} if there exists $\sigma \in \mathcal G$ such that $\C'=\sigma(\C)$.


The central objects of this paper are minimal codes, which are defined as follows.

\begin{definition}
 A linear code $\C$ is said to be \emph{minimal} if, for every $c, c' \in \C$,
$$\supp(c) \subseteq \supp(c') \, \, \Longleftrightarrow \,\,c=\lambda c' \mbox{ for some } \lambda \in \Fq.$$
\end{definition}


Moreover, we introduce the notion of reduced minimal code, which will allow us to study the maximal rates of minimal codes.

\begin{definition}
 An $[n,k]_q$ minimal code $\C$  is called \emph{reduced} if for every  $i \in \{1,\ldots, n\}$, the code $\C_i$ obtained by puncturing $\C$ on the coordinate $i$ (i.e. deleting the same coordinate $i$ in each codeword) is not minimal.
\end{definition}


\subsection{Projective systems}

In this section we consider linear codes from a geometrical view, as detailed in \cite{MR1186841}.  We first give some background of fundamentals of finite projective geometry. For a detailed introduction we refer to the recent book by Ball \cite{Ball2015}. Let $PG(k,q)$ be the finite
projective geometry of dimension $k$ and order $q$.  Due to a result of Veblen and Young
\cite{MR0179666}, all finite projective spaces of dimension greater than two are isomorphic, and they correspond to Galois geometries. The space $PG(k,q)$ can be easily seen as the vector space of dimension
$k+1$ over the finite field $\F_q$.  In this representation, the one-dimensional subspaces correspond to the
points, the two-dimensional subspaces correspond to the lines, etc. Formally, we have

$$ PG(k,q):= \left(\F_q^{k+1}\setminus \{0\}\right)/_\sim, $$
where
$$u\sim v \mbox{ if and only if } u=\lambda v \mbox{ for some } \lambda \in \F_q.$$

It is not hard to show
by elementary counting that the number of points of $PG(k,q)$ is given by \[\theta_q(k):=\frac{q^{k+1}-1}{q-1}.\]

A \textit{$d$-flat} $\Pi$ in $PG(k,q)$ is a subspace isomorphic to $PG(d,q)$; if $d=k-1$, the subspace $\Pi$ is called a \textit{hyperplane}. It is clear that $\theta_q(k)$ is also the number of hyperplanes in $PG(k,q)$.

Recall that a \emph{multiset $(\M, m)$ in $PG(k-1,q)$} is a set of points $\M \subseteq PG(k-1, q)$ together with a weight function $m$, which associates a positive integer $m(P)$ to all the points $P \in \M$. A multiset $(\M, m)$ is said to be \emph{finite} if $\sum_{P\in \M} m(P) < + \infty$.

\begin{definition}
Let $(\M,m)$ be a finite multiset in $ \Pi=PG(k-1,q) $. We define the \emph{character} function of $ \M $, denoted $ \ch_{\M} $, mapping the power set of $ \Pi $ to the non-negative integers:

 $$\ch_\M(A)=\sum_{P\in A} m(P).$$
So $\ch_{\M}(A) $ is  the number of points of $\M$ that belong also to $A$. With a slight abuse of notation, we will write $ m(P)=\ch_{\M}(P) $, for any point $ P $.
\end{definition}

 Central to the geometric point of view of linear codes is the idea of a projective system.

\begin{definition}
 A \emph{projective $[n,k,d]_q$ system} is a finite multiset $(\M,m)$  in  $PG(k-1,q)$, whose points do not lie all on a hyperplane, where $n=\sum_{P \in \M} m(P)$ , and $$d=n-\max\left\{ \ch_\M(H) \mid H \subset PG(k-1,q), \dim(H)=k-2\right\}.$$
Two  \emph{projective $[n,k,d]$ systems} $(\M,m)$ and $(\M',m')$ are said to be \emph{equivalent} if there exists a projective isomorphism $\phi$ of $PG(k-1,q)$ mapping $\M$ to $\M'$ which preserves the multiplicities, i.e. such that $m(P)=m'(\phi(P))$ for every $P \in \M$.
\end{definition}

Let $ \C $  be an $[n,k]_q$ code with $ k\times n $ generator matrix $ G $. Note that multiplying any column of $ G $ by a nonzero field element yields a generator matrix for a code which is equivalent to $ \C $. Consider the (multi)set of one-dimensional subspaces of $ \F_q^n $ spanned by the columns of $ G $. In this way the columns may be considered as a multiset $ (\M,m) $ of points in $ PG(k-1,q) $, where the weight function $m$ keeps track of how many times a certain column appears in the generator matrix, up to scalar multiple.

For any nonzero vector $ v=(v_1,v_2,\ldots,v_k) $ in $ \F_q^k $, it follows that the projective hyperplane
\[
v_1x_1+v_2x_2+\cdots + v_kx_k=0
\]
 contains $ |\M|-w $ points of $ \M $ if and only if the codeword $ vG $ has weight $ w $. Therefore,   linear non-degenerate  $ [n,k,d]_q $ codes and projective $ [n,k,d]_q $  systems are equivalent objects. Indeed, the procedure described above gives a  correspondence between  $[n,k,d]_q$ codes up to (monomial) equivalence and projective $[n,k,d]_q$ systems up to (projective) equivalence \cite[Theorem 1.1.6]{MR1186841}. This can be formally stated as follows. We denote by $(\Phi,\Psi)$ the correspondence
 $$\{ \mbox{ classes of non-deg. } [n,k,d]_q \mbox{ codes }\} \longleftrightarrow  \{ \mbox{ classes of projective } [n,k,d]_q \mbox{ systems }\}. $$
 More specifically, for a class of non-degenerate $[n,k,d]_q$ code $[\C]$, $\Phi([C])$ is the (equivalence class of the) multiset obtained by taking the columns with multiplicities of any generator matrix of any representative of $[C]$, while $\Psi$ is the functor that does the inverse operation. Given an equivalence class of  multisets $[(\M,m)]$ in $PG(k-1,q)$, it returns the class containing the code whose generator matrix has the points of $\M$, taken with multiplicities, as columns.
 It is not difficult to see that $(\Phi,\Psi)$ is an equivalence of the two categories (see \cite{assmus1998category} for a detailed discussion on the category of linear codes).
 
\section{Cutting blocking sets and minimal codes}\label{sec:cuttingsets}

Cutting blocking sets have been introduced by Bonini and Borello in \cite{bonini2019minimal}, for the construction of a particular family of minimal codes. However, we will show that minimal codes and cutting blocking sets are the same objects, under the equivalence $(\Phi, \Psi)$  between (non-degenerate) linear codes and  projective systems.

First we recall some basic background on blocking sets.

\begin{definition}Let $t, r, N$ be positive integers with $r<N$. A
\emph{$t$-fold $r$-blocking set} in $PG(N,q)$ is a set $\M\subseteq PG(N,q)$ such that for every $(N-r)$-flat $\Lambda$ of $PG(N,q)$ we have $|\Lambda \cap \M|\geq t$. When $r=1$, we will  refer to it as a \emph{$t$-fold blocking set}. When $t=1$, we will  refer to it as an \emph{$r$-blocking set}. Finally, \emph{blocking sets} are the ones with $r=t=1$.
\end{definition}

\begin{definition}
Let $r, N$ be positive integers with $r<N$. An $r$-blocking set $\M$ in $PG(N,q)$ is called \emph{cutting} if for every pair of $(N-r)$-flats $\Lambda, \Lambda'$ of $PG(N,q)$ we have
$$ \M \cap \Lambda \subseteq \M \cap \Lambda' \,\, \Longleftrightarrow \,\, \Lambda =\Lambda'.$$

Moreover, a cutting $r$-blocking set $\M$ is called \emph{minimal} if for every $P \in \M$, the set $\M\setminus\{P\}$ is not a cutting $r$-blocking set.
\end{definition}

The following result gives a different characterization of cutting blocking sets. The result follows also from \cite[Theorem 3.5]{bonini2019minimal}.

\begin{proposition}\label{prop:charactcutting}
 A set $\M \subseteq PG(N,q)$ is a cutting $r$-blocking set if and only if for every $(N-r)$-flat $\Lambda$ of $PG(N,q)$ we have $\langle \M \cap \Lambda \rangle =\Lambda$. \\
In particular, a cutting $r$-blocking set in $PG(N,q)$ is an $(N-r+1)$-fold blocking set.
\end{proposition}

\begin{proof}
\begin{itemize}
\item[($\Leftarrow$)] Let $\Lambda, \Lambda'$ be $(N-r)$-flats of $PG(N,q)$, such that $\M \cap \Lambda \subseteq \M\cap \Lambda'$. Then $\Lambda= \langle \M \cap \Lambda \rangle \subseteq \langle \M \cap \Lambda' \rangle = \Lambda'$, and since $\Lambda$ and $\Lambda'$ have the same dimension, we get $\Lambda=\Lambda'$, i.e. $\M$ is a cutting $r$-blocking set.
\item[($\Rightarrow$)] Suppose by contradiction that there exists an $(N-r)$-flat $\Lambda$ such that $\langle \Lambda \cap \M \rangle = \Delta \subsetneq \Lambda$. Then, for every $(N-r)$-flat $\Lambda'$ containing $\Delta$ we have
$\Lambda'\cap \M \supseteq  \Delta \cap \M  =\Lambda \cap \M$. And therefore, $\M$ is not a cutting $r$-blocking set.
\end{itemize}
\end{proof}

\begin{theorem}\label{thm:correspondence}
 Equivalence classes of $[n,k,d]_q$ minimal codes are in correspondence with equivalence classes of projective $[n,k,d]_q$  systems $(\M,m)$ such that $\M$ is a cutting blocking set via $(\Phi,\Psi)$. \\
 Furthermore, via the same pair of functors $(\Phi,\Psi)$, equivalence classes of $[n,k,d]_q$ reduced minimal codes are in correspondence with projective $[n,k,d]_q$  systems $(\M,m)$ such that $\M$ is a minimal cutting blocking set and $m(P)=1$ for every $P \in \M$.
\end{theorem}

\begin{proof}
 The first statement follows from the definitions of the two objects. Hyperplanes $\langle v \rangle^\perp$ in $PG(k-1,q)$ correspond to linearly independent codewords $vG$ of $\C$. For any pair of hyperplanes $H=\langle v \rangle^\perp$ and $H'=\langle v' \rangle^\perp$ we have $\M \cap H\subseteq \M \cap H'$ if and only if $\supp(vG)\supseteq \supp(v'G)$, where $G$ is any generator matrix of $\C$ and $(\M,m)$ is the associated projective system.
 
 Moreover, since puncturing on a coordinate of a code whose generator matrix is $G$ coincides to removing the corresponding point from the multiset $(\M,m)$, we get the second statement.
\end{proof}

Observe that reduced minimal codes correspond to multisets $(\M,m)$ with no multiplicity, i.e. such that $m(P)=1$ for every $P \in \M$. In particular, in order to construct minimal codes, by Theorem \ref{thm:correspondence} we only need to construct classical sets, without multiplicity. Therefore, from now on we will drop the multiplicity map from the notation when not necessary, and we will only talk about sets $\M \subseteq PG(N,q)$.

\section{Bounds on length and distance of minimal codes}\label{sec:bounds}

It is natural to ask for which values $R$ we can produce minimal codes of rate $R$. It is in general easier to construct minimal codes
with very small rate, such as symplex codes or related codes as in \cite{bartoli2019minimalLin,bonini2019minimal}. However, a priori it is not clear if one can do it for arbitrary rates. In particular, for a given dimension $k$ one would like to determine what is the smallest length $n$ (and hence the largest rate $R=k/n$) such that an $[n,k]_q$ code exists. In this section we provide some partial answers to these questions, proving some bounds on the length and the minimum distance of a minimal code for a fixed dimension. The characterization given in Theorem \ref{thm:correspondence} plays a crucial role in dealing with these problems.

\begin{theorem}\label{prop:lowboundgeo}
Let $\C$ be an $[n,k]_q$ minimal code. Then
$$n \geq (k-1)q+1.$$
\end{theorem}

\begin{proof}
If $k=1$ there is nothing to prove, hence we assume $k\geq 2$. Choose a generator matrix, and the corresponding  projective $ [n,k]_q $  system $(\M,m)$  in $ \Pi=PG(k-1,q)$.
 Consider the set $ S $ of incident point-hyperplane pairs $ (P,\Lambda) $ in $\Pi$, where $ P\in \M $. Summing over all the points of $ \M $ we obtain
\begin{equation}\label{eq1}
|S| = \sum_{P\in \M} m(P) \theta_q(k-2) = n  \theta_q(k-2),
\end{equation}
since $\theta_q(k-2)$ is the number of hyperplanes through a point.\\
On the other hand, summing over the set $\Gamma$ of all the hyperplanes of $ \Pi $ we get

\begin{equation}\label{eq2}
|S| =\sum_{H\in \Gamma} \ch_\M(H) \ge \sum_{H \in \Gamma} (k-1)= (k-1) \theta_q(k-1),
\end{equation}
where the inequality follows from the fact that $(\M, m)$  is in particular a $(k-1)$-fold blocking set in $\Pi$, by Proposition \ref{prop:charactcutting}.
Combining   \eqref{eq1} and \eqref{eq2}, we obtain
$$n \geq \left\lceil (k-1)\frac{\theta_q(k-1)}{\theta_{q}(k-2)} \right\rceil, $$
We then conclude observing that $\left\lceil (k-1)\frac{\theta_q(k-1)}{\theta_{q}(k-2)} \right\rceil=(k-1)q+\left\lceil\frac{k-1}{\theta_{q}(k-2)} \right\rceil=(k-1)q+1$.
\end{proof}

As a consequence, we get an asymptotic improvement of a result by Chabanne, Cohen and  Patey \cite{chabanne2013towards}. In that work, they showed that the rate $R$ of an  $[n,Rn]_q$ minimal code for $n$ large enough satisfies $R\leq \log_q(2)$, calling this bound the \emph{Maximal bound}.

\begin{corollary}
If $\C$ is a minimal code of rate $R$, asymptotically it holds $R\leq\frac{1}{q}$.
\end{corollary}

\begin{proof}
Let $\C$ be a minimal code of rate $R$. Then, by Theorem \ref{prop:lowboundgeo}
$$R=\frac{k}{n} \leq \frac{n+q-1}{qn} \longrightarrow \frac{1}{q},$$
as $n$ goes to infinity.
\end{proof}

We now prove an important result relating the minimum distance with the dimension of a minimal code and the size of the underlying field. We will give two different proofs of the theorem, to document further the interest of the geometric characterization.

\begin{theorem}\label{thm:dgeqkq2}
Let $\C$ be an $[n,k,d]_q$  minimal code with $k\geq 2$. Then $d\geq k+q-2$.
\end{theorem}

\begin{proof}
 Consider the projective $[n,k,d]_q$  system $(\M,m)$ associated to $\C$. Without loss of generality we can assume that there are no multiplicities, i.e. that $\M=\{P_1,\ldots, P_n\}$ with the $P_i$'s pairwise distinct. By Theorem \ref{thm:correspondence}, $\M$ is a cutting blocking set and there exists an hyperplane $H$ such that $\{ P_{d+1},\ldots, P_n\} \subseteq H$ and $P_1,\ldots, P_d \notin H$. Consider the set $\M^\prime:=\{ P_1,\ldots, P_d \}$. First we prove that $\M^\prime$ is a projective system, i.e. that $P_1,\ldots, P_d$ do not belong to the same hyperplane. Indeed, suppose that there exists a hyperplane $K$ such that $ P_1,\ldots, P_d \in K$, then clearly $K \neq H$, and hence $\Lambda:=H \cap K$ is a $(k-3)$-flat. Since there are $q+1$ hyperplanes containing $\Lambda$, there always exists a third hyperplane $T$  different from $H$ and $K$ such that $\Lambda \subseteq T$. Moreover $P_1,\ldots, P_d \notin T$, otherwise we would have $T=K$. Thus, we get $$\M \cap H=\{P_{d+1},\ldots, P_n\} \supseteq  \M \cap T,$$ which contradicts the fact that $\M$ is cutting. Therefore,  $\M^\prime$ is a projective $[d,k,d^\prime]_q$  system. 
 
 We show  now that $d^\prime\geq q-1$. Up to reordering the points, this means that there exists an hyperplane $H^\prime$ such that $P_1,\ldots, P_{d^\prime} \notin H^\prime$ and $P_{d^\prime+1},\ldots, P_d \in H^\prime$.  Consider the $(k-3)$-flat $\Lambda:=H \cap H^\prime$, and the sheaf of hyperplanes containing $\Lambda$. Except from $H$ and $H^\prime$ there are $q-1$ hyperplanes left in this sheaf. Clearly $P_{d^\prime+1},\ldots, P_d \notin \Lambda$, and hence they do not belong to any of the remaining $q-1$ hyperplanes. Moreover, every point in $\{P_1,\ldots, P_{d^\prime}\}$ can be in at most one hyperplane of the sheaf. Assume by contradiction that $d^\prime\leq q-2$, then there exists at least a hyperplane $\tilde{H}  \neq H$ such that
 $P_1,\ldots, P_d \notin \tilde{H}$. Hence 
 $$\M \cap H=\{P_{d+1},\ldots, P_n\} \supseteq  \M \cap \tilde{H},$$
 which contradicts the fact that $\M$ is cutting. Thus, $\M^\prime$
 is a projective $[d,k,d^\prime]_q$  system, with $d'\geq q-1$. Combining it with the Singleton bound \cite{singleton} we obtain
 $$ d \geq k+d^\prime-1\geq k+q-2.$$
\end{proof}

An alternative proof, given from a coding theory point of view, is the following.

\begin{proof}[Second proof.]
 Let $c$ be a codeword of minimum weight $d$ in $\C$. Then consider the code $\C^\prime$ obtained by puncturing $\C$ in all the $n-d$ coordinates where $c$ is $0$. Observe that $\C^\prime$ is a $[d,k]_q$ code. Indeed, if the dimension of $\C^\prime$ is less than $k$, it means that there is at least one codeword $w$ in $\C$ whose support is disjoint from the support of $c$. Hence $\supp(c+w)$ contains $\supp(c)$ and $\supp(w)$ and this contradicts the minimality of $\C$.
 
 Now, observe that $\C^\prime$ has distance $d^\prime \geq q-1$. Indeed, consider $c^\prime\in\C^\prime$ that corresponds to $c$ and  has weight $d$, and  let $u \in \C$ such that the corresponding $u^\prime \in \C^\prime$ is of minimum weight $d^\prime$ in $\C^\prime$. Then, for any $\alpha\in\Fq^\ast$, consider the codeword $c^\prime+\alpha u^\prime$ in $\C^\prime$. If $d^\prime< q-1$, at least one of these codewords has weight $d$. The corresponding codeword in $\C$, then, has support containing $\supp(c)$, which yields a contradiction to the minimality of $\C$.
 
 Finally, we apply the Singleton bound on $\C^\prime$ and combine it with $d^\prime \geq q-1$ to obtain the desired result:
 
$$d\geq d^\prime +k -1 \geq k+q-2.$$
\end{proof}

\begin{remark}\label{rem:minimumdistance}
The  bound in Theorem \ref{thm:dgeqkq2} is not sharp in general: considering the second proof, we remark that $d=k+q-2$ if and only if $\C^\prime$ is a $[q+k-2,k,q-1]_q$ MDS code with exactly $q-1$ codewords of weight equal to the length (namely, all the nonzero multiples of $c^\prime$). Weight enumerators of MDS codes are known (see for example \cite[Ch. 11, \S 3, Theorem 6]{huffmanpless}), so that it is easy to prove that this may happen if and only if 
 $$\sum_{j=0}^{k-1}(-1)^j\cdot \binom{k-1+q-2}{j}\cdot q^{k-1-j}=1$$
 which is not true 
 for $q\neq 2$ and $k\geq 3$. Moreover, one can also observe that assuming the MDS conjecture to be true (see \cite{segre1955curve}), a $[q+k-2,k,q-1]_q$ MDS code exists only for $k \leq 3$. 
 
\end{remark}

As a result, we can actually get new bounds on the length of a minimal code, combining Theorem  \ref{thm:dgeqkq2} with known upper bounds on the minimum distance.
It is easy to observe that using the Singleton bound does not improve on Theorem \ref{prop:lowboundgeo}. However, if $q$ is small, we can get better results using the Griesmer bound \cite{griesmer}.

\begin{corollary}\label{coro:griesmer} Let $\C$ be an $[n,k]_q$ minimal code. Then
$$n\geq \sum _{i=0}^{k-1}\left\lceil {\frac {k+q-2}{q^{i}}}\right\rceil. $$
\end{corollary}

\begin{proof}
It follows combining Theorem \ref{thm:dgeqkq2} with the Griesmer bound.
\end{proof}

\begin{remark}
Observe that for some sets of parameters Corollary \ref{coro:griesmer} gives a better lower bound on the length of minimal codes than the one of Theorem \ref{prop:lowboundgeo}, while for other sets of parameters the converse holds. For instance, it is easy to see that for $q=2$, Corollary \ref{coro:griesmer} is always better. Viceversa, when $q\geq k \geq 4$, Theorem \ref{prop:lowboundgeo} provides better results.

Furthermore, numerical results with {\sc Magma} show that the bound in Corollary \ref{coro:griesmer} is not sharp. For example, for $q=2$ and $k=4$, the minimum possible length of a minimal code is $9$, while the above bound gives $8$.
\end{remark}

\subsection{Asymptotic performance of minimal codes}

We recall that there is an existence result that holds asymptotically, i.e. we can actually ensure the existence of minimal codes of arbitrary length $n$ of a fixed rate $R$ that only depends on $q$. This existence result is not constructive, and it was shown by Chabanne, Cohen and  Patey \cite{chabanne2013towards}.

\begin{theorem}[Minimal Bound \cite{chabanne2013towards}] \label{thm:minimalbound}
For any rate $R=k/n$ such that $$0 \leq R \leq \frac{1}{2}\log_q \left(\frac{q^2}{q^2-q+1} \right),$$  there exists an infinite sequence of $[n, k]_q$ minimal codes.
\end{theorem}

The most important consequence of Theorem \ref{thm:dgeqkq2} is that it allows to show that minimal codes are asymptotically good. Let us recall that a family of codes is said \emph{asymptotically good} if it contains a sequence $C=(\C_1,\C_2, \dots)$ of linear codes, where $C_n$ is an $[n,k_n,d_n]_q$ code such that the rate $R$ and the relative distance $\delta$ of $C_n$, that is
$$R:= \liminf_{n\to \infty} {\frac{k_n}{n}} \ \ \ \ \text{ and } \ \ \ \ 
    \delta := \liminf_{n\to \infty} {\frac{d_n}{n}},
$$
are both positive.

In general, we would like ideally both rate and relative distance of a code to be as large as possible, since the rate measures the number of information coordinates with respect to the length of the code and the relative distance measures the error correction capability of the code. Determining the rate and the relative distance for a class of codes is in general a difficult task. For example, it is still unknown if the family of cyclic codes is asymptotically good. 
However, some families of asymptotically good codes are known to exist. For example, codes that meet the Asymptotic Gilbert-Varshamov bound, binary quasi-cyclic codes \cite{chen1969some,alahmadi2017long}, self-dual codes \cite{alahmadi2018self}, group codes \cite{borello2019group}.

A direct consequence of Theorem \ref{thm:dgeqkq2} and of the Minimal Bound of Theorem \ref{thm:minimalbound} is the following result.

\begin{theorem}\label{thm:asymgood}
Minimal codes are asymptotically good.
\end{theorem}






\subsection{Cutting blocking sets in the projective plane}

In the projective plane, we can get better bounds on the cardinality of cutting blocking sets. This is due to the following result, which shows that we can reduce to study the cardinality of $2$-fold blocking sets.

\begin{lemma}\label{lem:eqcutting2fold}
In $PG(2,q)$ a set $\M$ is a cutting blocking set if and only if it is a $2$-fold blocking set.
\end{lemma}

\begin{proof}
Clearly, a cutting blocking set is a $2$-fold blocking set, as shown in Proposition \ref{prop:charactcutting}. On the other hand, if $\M$ is a $2$-fold blocking set, then for every line $\ell$ in $PG(2,q)$ $\langle \ell \cap \M \rangle =\ell$, since $|\ell \cap \M | \geq 2$. We conclude again by Proposition \ref{prop:charactcutting}.
\end{proof}

Using this equivalence, we can give upper bounds on minimal cutting blocking sets and lower bounds on cutting blocking sets. Thanks to the correspondence of Theorem \ref{thm:correspondence} between minimal codes and cutting blocking sets, we can regard these bounds as bounds on the length of minimal codes of dimension $3$. 

In particular, the following theorems follow directly  from Lemma \ref{lem:eqcutting2fold} and results in \cite{ball1996size, bishnoi2018minimal}.

\begin{theorem}[\cite{ball1996size}] \label{thm:mink3} Let $\C$ be an $[n,3]_q$ minimal code.
\begin{enumerate}
\item If $q<9$, then $n \geq 3q$.
\item If $q \in \{11,13,17,19\}$, then $n \geq (5q+7)/2$.
\item If $q>19$ and $q=p^{2d+1}$ for some $p$ prime and $d \in \mathbb N$, then $n \geq p^d\left\lceil \frac{p^{d+1}+1}{p^d+1} \right\rceil +2$.
\item If $q>4$ and $q$ is a square, then $n\geq 2q+2\sqrt{q}+2$.
\end{enumerate}
\end{theorem}

\begin{theorem}[\cite{bishnoi2018minimal}]\label{thm:minredk3}
Let $\C$ be an $[n,3]_q$ reduced minimal  code. Then
$$n \leq \frac{q}{2} \left(\sqrt{ 8q-7}+ 1 \right)+2.$$
\end{theorem}

\begin{remark}
Observe that the minimal codes of Theorem \ref{thm:mink3} correspond to cutting blocking sets in $PG(2,q)$ and the reduced minimal codes of Theorem \ref{thm:minredk3} correspond to minimal cutting blocking sets in $PG(2,q)$, via the correspondence $(\Phi,\Psi)$ of Theorem \ref{thm:correspondence}.
\end{remark}

It would be interesting to have similar results for projective spaces of larger dimension, but in this case the equivalence of Lemma \ref{lem:eqcutting2fold} does not hold.

\color{black}

\section{Construction of minimal codes}\label{sec:constructions}

In this section we provide a general construction of reduced minimal codes based on the geometric point of view. For this family of codes, we also determine the weight distribution, using basic combinatorial results in finite geometry.

We start with two  auxiliary lemmas, based on avoiding results in finite projective spaces.

\begin{lemma}\label{lem:avoidinghyper}
 Let $q$ be a prime power, $k, r$ be integers such that $1\leq r \leq k$. Let $P_1,\ldots, P_r \in PG(k-1,q)$ be points not on the same $(r-2)$-flat. Then, the number of hyperplanes $H$ avoiding $P_1,\ldots,P_r$ is 
  $q^{k-r}(q-1)^{r-1}.$
\end{lemma}

\begin{proof}
It follows from a simple calculation using inclusion-exclusion principle. Since the number of hyperplanes is $\theta_q(k-1)$, and the number of hyperplanes containing at least $i$ points among the $P_j$'s is equal to $\theta_q(k-1-i)$, we get that the number of hyperplanes avoiding all the $P_j$'s is  
 \begin{align*}
     & \; \theta_q(k-1)-\sum_{i=1}^{r}(-1)^{i-1}\binom{r}{i}\theta_q(k-1-i) \\
     = & \; \frac{1}{q-1}\sum_{i=0}^r\binom{r}{i}(-1)^i(q^{k-i}-1) \\
     = & \;  \frac{1}{q-1}\left(q^{k-r} \sum_{i=0}^r \binom{r}{i}(-1)^iq^{r-i} -\sum_{i=0}^r \binom{r}{i}(-1)^i \right)\\
     = & \; q^{k-r}(q-1)^{r-1}. 
 \end{align*} 
\end{proof}

\begin{lemma}\label{lem:countnumbers}
 Let $P_1,\ldots, P_k \in PG(k-1,q)$ be points in general position. Then, the number of hyperplanes  containing $P_1,\ldots, P_s$ and avoiding $P_{s+1},\ldots, P_k$ is $(q-1)^{k-s-1}$ 
\end{lemma}

\begin{proof}
 Let $\Lambda:=\langle P_{1}, \ldots, P_s\rangle$, then the number of hyperplanes containing $\Lambda$ and avoiding $P_{s+1},\ldots, P_{k}$ is in correspondence with the number of hyperplanes in $ PG(k-1)/\Lambda \cong PG(k-s,q)$ avoiding $P_{s+1},\ldots, P_k$. Such number is, by Lemma \ref{lem:avoidinghyper}, equal to $(q-1)^{k-s-1}$.
\end{proof}

\begin{theorem}\label{thm:tetraedro}
 Let $P_1,\ldots, P_k$ be points in general position in $PG(k-1,q)$. For $0 \leq i <j \leq k$, consider the line $\ell_{i,j}:=\langle P_i, P_j \rangle.$ Then, $\M:=\bigcup_{i,j} \ell_{i,j}$ is a minimal cutting blocking set.
\end{theorem}

\begin{proof}
Let $H$ be a hyperplane in $PG(k-1,q)$. Since the points $P_1,\ldots, P_k$ are in general position, there exists at least one point among them, say $P_1$ that is not in $H$. Consider the intersection $H \cap \M$, which does not contain $P_1$. Hence $H$ meets the lines $\ell_{1,j}$'s in $k-1$ distinct points $Q_2,\ldots Q_k$, i.e. $\{Q_j\}=H \cap \ell_{1,j}$ for $j \in \{2,\ldots, k\}$. Take the flat $\Lambda:= \langle \M \cap H\rangle$, and observe that $$\langle \Lambda, P_1\rangle\supseteq \langle P_1,Q_j\rangle =\ell_{1,j},$$
However, $P_j \in \ell_{1,j}$ for every $j\in \{2,\ldots, k\}$, and this implies $\langle \Lambda, P_1\rangle \supseteq \langle P_1,\ldots, P_k \rangle =PG(k-1,q)$. Hence, necessarily $\dim(\Lambda)=k-2$ and by Proposition \ref{prop:charactcutting}, $\M$ is a cutting blocking set.

It is left to prove that $\M$ is minimal. Suppose we remove from $\M$ one of the points $P_i$'s from $\M$, say $P_1$, getting $\tilde{\M}:=\M \setminus \{P_1\}$. Take a $(k-3)$-flat  $\Lambda \subseteq \langle P_2,\ldots, P_k\rangle$ avoiding the points $P_2,\ldots, P_k$. By Lemma \ref{lem:avoidinghyper}  such an hyperplane always exists. Hence $H:=\langle \Lambda, P_1 \rangle$ is an hyperplane such that $H\cap \tilde{\M}\subseteq \Lambda$, and by Proposition \ref{prop:charactcutting}, $\tilde{M}$ is not minimal. Similarly, choose a point in $\M\setminus\{P_1,\ldots, P_k\}$ and remove it from $\M$. Without loss of generality, we can choose $Q_{1,2} \in \ell_{1,2}\setminus\{P_1,P_2\}$ and consider $\tilde{M}:=\M \setminus \{Q_{1,2}\}$.  Take the space $H:=\langle Q_{1,2}, P_3,\ldots, P_k\rangle$. It is easy to see that $\langle H,P_1\rangle =\langle H,P_2 \rangle =PG(k-1,q)$, and hence $H$ is an hyperplane. Moreover, $H\cap \M=\{Q_{1,2}, P_3,\ldots, P_k\}$, therefore $\dim (H \cap \tilde{\M})=\dim (\langle P_3,\ldots, P_k \rangle)=k-3$, and by Proposition \ref{prop:charactcutting} $\tilde{\M}$ can not be a cutting blocking set. 
\end{proof}

The next result analyzes the reduced minimal code obtained in Theorem \ref{thm:tetraedro}, giving the full description of its weight distribution.

\begin{theorem}\label{thm:tetraedrocode}
 The code associated to the minimal cutting blocking set of Theorem \ref{thm:tetraedro} is a $[\binom{k}{2}(q-1)+k,k]_q$ reduced minimal code $\C$, whose weights are exactly
 $$f_{q,k}(r):= \frac{1}{2}(k-r)((k+r-1)q-2k+4),$$
 for every $r \in \{0,\ldots, k-1\}$.
 Furthermore, the weight distribution of $\C$ is given by
 $$A_i(\C)=\sum_{\{r\mid f_{q,k}(r)=i\}}\binom{k}{r}(q-1)^{k-r}.$$
\end{theorem}

\begin{proof}
 By the equivalence $(\Phi,\Psi)$ between codes and projective systems, the dimension of the code obtained by $\M$ is clearly $k$ and its length is $n=\binom{k}{2}(q-1)+k$. 
  Now, for an hyperplane $H=\langle v \rangle^{\perp}$, the weight
 of its $q-1$ associated codewords (i.e. all the nonzero multiples of $vG$, where $G$ is the generator matrix obtained from $\M$) is $n-|\M \cap H|$. Therefore, it is determined by  $|H \cap \M|$. By the symmetric properties of $\M$, the quantity $|H \cap \M|$ only depends on the integer 
 $$r:=|\{i \in \{1,\ldots, k\} \mid P_i \in H\}|.$$ 
 In this case, without loss of generality we can assume that $P_1,\ldots, P_r \in H$, and $P_{r+1},\ldots,$ $P_{k-r} \notin H$. Hence, $\M$ contains all the lines $\ell_{i,j}$ for $0 \leq i <j \leq r$, and it intersects all the lines $\ell_{i,j}$ in $\{P_i\}$, for $0 \leq i \leq r<j \leq k$ , and in $\{Q_{i,j}\}$ for $r+1 \leq i <j  \leq k$. Moreover, observe that the points $Q_{i,j}$ are all pairwise distinct. Therefore, the weight of the codeword associated to $H$ is equal to 
 \begin{align*}
    f_{q,k}(r) = & \;\binom{k}{2}(q-1)+k-|\M \cap H| \\
     = & \; \binom{k}{2}(q-1)+k-\Big|\bigcup_{0 \leq i \leq r<j \leq k}\ell_{i,j} \Big| - \Big|\bigcup_{r+1 \leq i <j  \leq k}\{Q_{i,j}\} \Big| \\
     = & \; \binom{k}{2}(q-1)+k-\binom{r}{2}(q-1)-r+\binom{k-r}{2} \\
    = & \; \frac{1}{2}(k-r)((k+r-1)q-2k+4).
 \end{align*}
 The numbers $A_{i}(\C)$ follow from Lemma \ref{lem:countnumbers}, taking into account that for every hyperplane we need to count  $q-1$ distinct codewords, which correspond to all the nonzero multiples.
\end{proof}

\begin{example}
 We explain now in details the situation for $k=3$. The construction of the minimal cutting blocking set of Theorem \ref{thm:tetraedrocode} corresponds to the union of three lines $\ell_1,\ell_2,\ell_3$ in the projective plane $PG(2,q)$ with trivial intersection, that is $\ell_1 \cap \ell_2 \cap \ell_3 = \emptyset$. We write   $\{P_{i,j}\} =\ell_i \cap \ell_j$ for  $1\leq i <j \leq 3$. Here hyperplanes are lines and for any line $\ell$ there are three possibilities: it can coincide with one of the lines $\ell_i$'s, it can contain one of the $P_{i,j}$'s, or none of them. The three cases give weights 
 $f_{q,3}(2)=2q-1$, $f_{q,3}(1)=3q-2$  and $f_{q,3}(0)=3q-3$. This code for $q \geq 3$ is a three-weight code with weight distribution $A_0=1$, $A_{2q-1}=3(q-1)$,  $A_{3q-3}=(q-1)^3$ and  $A_{3q-2}= 3(q-1)^2$,  and for $q=2$ it is a two-weight code with weight distribution $A_0=1$, $A_{3}=4$ and  $A_{5}=3$.
\end{example}

\begin{remark}
The family of codes described in Theorem \ref{thm:tetraedro} has  been constructed independently also in \cite{bartoli2019inductive}. However, in that paper the authors provided only the construction for $q \geq k+2$ and they did not study the reducedness, nor find the weight distributions. 
This suggests that the geometric point of view allows to analyze better the properties of minimal codes.
\end{remark}

\begin{remark}
 The construction of Theorem \ref{thm:tetraedrocode} for dimension $k=3$ gives rise to minimal codes of shortest possible length, whenever $q<9$. This follows from Theorem \ref{thm:mink3}.
 It is not clear, however, if this is true also when $q\geq 9$.
\end{remark}

\subsection{Minimal codes of dimension $4$}\label{subse:dim4}
Here we exhibit a special construction for minimal codes of dimension $4$, using cutting blocking sets in $PG(3,q)$ which have size smaller than the ones provided in Theorem \ref{thm:tetraedro}.

\begin{construction}\label{constrdim4}
 Let $ P_1,P_2,P_3,P_4\in PG(3,q)$ be points in general position. Up to change of coordinates, we can assume them to be the (representatives of the) standard basis vectors. Consider the lines $\ell_i=\langle P_i, P_{i+1}\rangle$ for $i \in \{1,2,3,4\}$ and the indices taken modulo $4$. For the line $m_1:=\langle P_1,P_3\rangle$, consider the sheaf of planes  $\{H_{\alpha} \mid \alpha  \in \Fq^*\}$ containing it, given by $H_{\alpha}:=\{[x,y,z,\alpha y]\mid [x,y,z]\in PG(2,q)\}$, where we have removed the planes $\langle \ell_1, \ell_2 \rangle$ and $\langle \ell_3,\ell_4 \rangle$.
 For the line $m_2:= \langle P_2, P_4 \rangle$, we do the same, and take  the sheaf of planes $\{K_{\alpha} \mid \alpha  \in \Fq^*\}$ containing it, given by $K_{\alpha}:=\{[x,y,\alpha x,z]\mid [x,y,z]\in PG(2,q)\}$, where we have removed the planes $\langle \ell_1, \ell_4 \rangle$ and $\langle \ell_2,\ell_3 \rangle$.
 Now, for every $\alpha \in \Fq^\ast$ compute $H_\alpha \cap K_\alpha=\{[x,y,\alpha x, \alpha y] \mid [x,y] \in PG(1,q)\}$.  We fix a $\beta \in \Fq^\ast$, and take the point 
 $$Q_{\beta,\alpha}:=[1,\beta,\alpha,\beta \alpha].$$
 Note that  $Q_{\beta,\alpha}\in (H_\alpha \cap K_\alpha)\setminus (m_1 \cup m_2)$ for every $\alpha \in \Fq^\ast$. Moreover, the points $Q_{\beta,\alpha}$ are all on the line $\ell_{\beta}:=\langle [1,\beta,0,0],[0,0,1,\beta]\rangle=\{[x,\beta x, y, \beta y] \mid [x,y] \in PG(1,q)\}$.

 With this notation we define $\M_{\beta}$ to be the set
 $$\M_{\beta}:=\ell_1 \cup \ell_2 \cup \ell_3 \cup \ell_4 \cup \{Q_{\beta,\alpha} \mid \alpha \in \Fq^\ast\}.$$
 
\end{construction}

\begin{theorem}\label{thm:constructionP3}
The set $\M_{\beta}$ is a minimal cutting blocking set in $PG(3,q)$, for every $\beta \in \Fq^\ast$.
\end{theorem}

\begin{proof}
Let $H$ be a hyperplane of $PG(3,q)$, that is a plane. We call $\mathcal N$ the union of the four lines.
First, it is easy to see that if $H$ contains a line $\ell_i$, then $\langle H\cap \M \rangle$ is an hyperplane, since it contains at least another point not on $\ell_i$. 
Suppose that $H$ meets a line $\ell_{i}$ in only one point $R_{i}$ distinct from $P_i$ and $P_{i+1}$. Without loss of generality, we can assume $i=1$. Hence $\langle \M_{\beta}\cap H \rangle \supseteq \langle \mathcal N \cap H \rangle =:\Lambda$. Now, observe that $\langle \Lambda, P_1\rangle$ contains at least the line $\ell_{1}$, a point on $\ell_2$ distinct from $P_2$ and another point on $\ell_4$ different from $P_1$. Hence $$\langle \Lambda,P_1\rangle \supseteq \langle \ell_1,\ell_2,\ell_4\rangle \supseteq\langle P_1,P_2,P_3,P_4\rangle=PG(3,q),$$
which implies $\dim(\Lambda)=2$. It remains to analyze the only case left, which is $\mathcal N \cap H =\{P_1,P_3\}$ (the case $\mathcal N\cap H=\{P_2,P_4\}$ is symmetric). In this case, necessarily $H=H_{\alpha}$, for some $\alpha \in \Fq^\ast$, and so $\langle H\cap \M_{\beta}\rangle =\langle P_1,P_3,Q_{\beta,\alpha}\rangle=H_{\alpha}=H$. This shows that $\M_{\beta}$ is a cutting blocking set.

It remains to prove that $\M_{\beta}$ is minimal. Clearly, we can not remove any of the points $Q_{\beta,\alpha}$'s, since $\M_{\beta}\setminus \{Q_{\beta,\alpha}\}$ meets $H_\alpha$ only in $P_1$ and $P_3$. The same happens if we remove one of the points $P_i$'s. Indeed, $\M_{\beta}\setminus \{P_1\}$ meets $H_\alpha$ only in $\{P_3,Q_{\beta,\alpha}\}$, for every $\alpha \in\Fq^\ast$ (and symmetrically with $\M_{\beta}\setminus\{P_3\}$). The same happens with the hyperplanes $K_\alpha$'s if we remove $P_2$ or $P_4$. It is left to prove that if we remove a point $R$ on one of the lines, say $\ell_1$, the resulting set $\M_{\beta}\setminus\{R\}$ is not cutting. Take the point $P_3$ and consider the sheaf of planes containing the line $\langle P_3,R\rangle$. Every plane of this sheaf meets the line $\ell_4$ in exactly one point. Hence, the sheaf is parametrized by the points on the line $\ell_4$, and we can write it as $\{H_S \mid S \in \ell_4\}$, where clearly $H_S=\langle P_3,R,S\rangle$. Consider now the intersection between $H_S$ and $\tilde{\mathcal N}:=\mathcal N\setminus\{R\}$, i.e. the union of all the four lines without the point $R$. If $S=P_1$ then $H_{P_1}\cap \tilde{\mathcal N}=(\ell_1\setminus \{R\})\cup \{P_3\}$, which spans a hyperplane. It is not difficult to see that in all the remaining $q$ cases it spans a line. However, every $H_S$ meets the line $\ell_{\beta}$ in exactly a point. Hence it  contains at most one of the $Q_{\beta,\alpha}$'s. However, we have $q$ hyperplanes $H_S$ and only $q-1$ points. Therefore, necessarily there exists $S\in \ell_4$ such that $H_S\cap \M_{\beta}=\{P_3,R,S\}$ and thus $\M_{\beta}\setminus\{R\}$ is not cutting.

\end{proof}

\begin{corollary}\label{cor:P3}
 For every $\beta \in \Fq^\ast$, Construction \ref{constrdim4}  produces  a
 $[5q-1,4,3q-2]_q$ reduced minimal code $\C_{\beta}$.
\end{corollary}

\begin{proof}
 Using the characterization result of Theorem \ref{thm:correspondence}, clearly the code obtained by the minimal cutting blocking set $M_{\beta}$ via $(\Phi,\Psi)$ is a $[5q-1,4]_q$ reduced minimal code. It is left to determine the minimum distance of $\C_{\beta}$, which corresponds via $(\Phi,\Psi)$ to the value $(5q-1)-\max\{|H\cap \M_{\beta}| : \dim(H)=2\}$. Any hyperplane $H$ can contain at most two of the lines $\ell_i$'s and $\ell_{\beta}$, since every three of them span the whole space $PG(3,q)$. If it contains none of them, then $|\M_{\beta}\cap H|\leq 5$. If $H$ contains only one of the $\ell_i$'s then $|\M_{\beta}\cap H|\leq q+3$. In the case $H$ contains only $\ell_{\beta}$ we also have $|\M_{\beta}\cap H|\leq q+3$. Finally, the only case in which $H$ contains a pair of lines is when $H=\langle \ell_i, \ell_{i+1}\rangle$, for $i \in \{1,2,3,4\}$ (where the indices are taken modulo $4$). In this case, we can see that $H$ does not contain any of the points $Q_{\alpha,\beta}$, and therefore, $|H\cap M_{\beta}|=2q+1$. For every prime power $q$, the maximum among these values is given by $2q+1$, and this concludes the proof. 
\end{proof}

We conclude this subsection with explanatory examples.


\begin{example}
 According to Corollary \ref{cor:P3}, Construction \ref{constrdim4} for $q=2$ and $\beta=1$ gives rise to a minimal $[9,4,4]_2$ code, whose generator matrix is
 $$G=\begin{pmatrix}
1 & 1 & 0 & 0 & 0 & 0 & 0 & 1 & 1 \\
0 & 1 & 1 & 1 & 0 & 0 & 0 & 0 & 1 \\
0 & 0 & 0 & 1 & 1 & 1 & 0 & 0 & 1 \\
0 & 0 & 0 & 0 & 0 & 1 & 1 & 1 & 1 \\
\end{pmatrix}.$$
 It was proved by computer search with {\sc Magma} \cite{magma} that 9 is the shortest length that a minimal code of dimension $4$ can have over $\F_2$. Moreover, always with {\sc Magma} we observed that this is the unique $[9,4]_2$ minimal code up to equivalence. 

\end{example}

\begin{example} 
For $q=3$ and $\beta=2$, Construction \ref{constrdim4} gives the $[14,4,7]_3$ reduced minimal code $\C_2$ whose generator matrix is 
$$G=\left( \begin{array}{cccccccccccccc}
1 & 1 & 1 & 0 & 0 & 0 & 0 & 0 & 0 & 0 & 1 & 2 & 1 & 1\\
0 & 1 & 2 & 1 & 1 & 1 & 0 & 0 & 0 & 0 & 0 & 0 & 2 & 2\\
0 & 0 & 0 & 0 & 1 & 2 & 1 & 1 & 1 & 0 & 0 & 0 & 1 & 2\\
0 & 0 & 0 & 0 & 0 & 0 & 0 & 1 & 2 & 1 & 1 & 1 & 2 & 1\\
\end{array} \right).$$
\end{example}

\subsection{Minimal codes of dimension $5$}\label{subse:dim5}
Here we show another special construction for minimal codes of dimension $5$, using cutting blocking sets in $PG(4,q)$ whose  size is smaller than the one provided in Theorem \ref{thm:tetraedro}. When $q=2$, we provide also an alternative construction for minimal codes of dimension $5$ as minimal blocking sets in $PG(4,2)$.

\begin{construction}\label{pentagon}
Let $P_1,P_2, P_3, P_4, P_5$ be five points in general position in $PG(4,2)$. Without loss of generality, we can assume that they are the (representatives of the) standard basis vectors. Consider the lines $\ell_i = \langle P_i, P_{i+1} \rangle $ for $i\in\{1,2,3,4,5\}$, where  the indices are taken modulo $5$. Consider now for $i \in \{1,2,3,4\}$ a point $Q_i \in \ell_i\setminus\{P_i,P_{i+1}\}$, and define the lines $m_1:=\langle Q_1, Q_3 \rangle$,  $m_2:=\langle Q_2, Q_4 \rangle$ and $m_3:=\langle Q_1, Q_4 \rangle$. 

With this notation, we define the set $\M:=\ell_1\cup\ell_2\cup\ell_3\cup\ell_4\cup\ell_5\cup m_1 \cup m_2 \cup m_3 $.
We will refer to the above construction also as the \emph{pentagonal construction}.
\end{construction}

\begin{theorem}
The set $\M$ defined  in Construction \ref{pentagon} is a  cutting blocking set in $PG(4,q)$. 
\end{theorem}

\begin{proof}
 We first write $\mathcal N=\ell_1\cup\ell_2\cup\ell_3\cup\ell_4\cup\ell_5$ and $\mathcal N^\prime=m_1 \cup m_2 \cup m_3$. Let $H$ be a hyperplane, define the spaces $\Lambda:=\langle H \cap \M\rangle$ and $\Lambda_1:=\langle H \cap \mathcal N\rangle$ and consider the number $r$ of the  $P_i$'s that are also in $H$. Clearly $r\in \{0,1,2,3,4\}$. If $r=4$ it is clear that $\langle H\cap \mathcal N\rangle=H$. If $r=0$, then it is easy to see that $\langle \Lambda_1,P_1\rangle$ contains $\mathcal N$, and hence it is the whole $PG(4,q)$. Therefore $\dim(\Lambda_1)=\dim(\Lambda)=3$. Also if $r=1$, that is $P_1\in H$, then  $\langle \Lambda_1,P_2\rangle$ turns out to be the whole space, hence $\dim(\Lambda)=3$. Now assume that $r=3$. Then we have two possibilities for the indices of these points. They can be consecutive (modulo $5$), say $P_1,P_2,P_3$, in which case $H$ contains their span plus a point on $\ell_4$. Clearly this implies $\dim(\Lambda)=\dim(\Lambda_1)=3$. The second case is when the indices are of the form $i,i+1,i+3$, i.e. $\langle \ell_i,P_{i+3}\rangle \subseteq H$. Then $H$ intersects at least one  line $m_j=\langle Q_t,Q_s\rangle$ skew to $\ell_i$ in another point $R$ distinct from $Q_1,Q_2,Q_3,Q_4$. Consider then $\langle \Lambda, Q_s \rangle\supseteq \langle \ell_i,m_j, P_{i+3}\rangle =PG(4,q)$. Hence also in this case $\dim(\Lambda)=3$. It remains to show the case $r=2$. If the indices of these two points are consecutive, then $H$ contains a line $\ell_i$ and two more points, one on $\ell_{i+2}$ and one on $\ell_{i+3}$. Clearly in this case $\langle \Lambda_1,P_{i+3} \rangle \supseteq \langle \ell_i, \ell_{i+2}, \ell_{i+3} \rangle=PG(4,q)$, and we conclude also in this case. Suppose now that the two points in $H$ are $P_i$ and $P_{i+2}$. Then $H$ will also intersect the line $\ell_{i+3}$ in a point $R$, and at least a line $m_j=\langle Q_t,Q_s\rangle$ in a point $S$, which is different from $Q_t$ and $Q_s$. Then it is  easy to see that also in this case $\langle \Lambda, Q_s \rangle=PG(4,q)$, which finally shows that $\M$ is cutting.
\end{proof}

\begin{corollary}
Construction \ref{pentagon} produces a $[8q-3,5,4q-3]_q$  minimal code.
\end{corollary}

\begin{proof}
 The fact that from Construction \ref{pentagon} we obtain a  $[8q-3,5]_q$  minimal code, simply follows from the characterization result of Theorem \ref{thm:correspondence}. The minimum distance can be computed observing that a hyperplane $H$ can contain at most $4$ lines among the defining lines of $\M$, and this happens only in five cases: $H_1=\langle \ell_1, \ell_2, m_2, m_3 \rangle$, $H_2=\langle \ell_3, \ell_4, m_1, m_3 \rangle$, $H_3=\langle \ell_1, \ell_2, \ell_3, m_1 \rangle$, $H_4=\langle \ell_2, \ell_3, \ell_4, m_2 \rangle$ and $H_5=\langle \ell_1, \ell_4, \ell_5, m_3 \rangle$. In these cases we have $|\M \cap H_i|=4q$, and the weights of the associated codewords are $8q-3-4q=4q-3$. In all the other cases, it is not difficult to see that any other hyperplane contains a smaller number of points of $\M$.  Hence, the minimum distance of the code is $4q-3$.
\end{proof}

In the binary case, the pentagonal construction gives the $[13,5,5]_2$ reduced minimal code whose generator matrix is

 $$G=\left(\begin{array}{ccccccccccccc}
1 & 1 & 0 & 0 & 0 & 0 & 0 & 0 & 0 & 1 & 1 & 0 & 1 \\
0 & 1 & 1 & 1 & 0 & 0 & 0 & 0 & 0 & 0 & 1 & 1 & 1 \\
0 & 0 & 0 & 1 & 1 & 1 & 0 & 0 & 0 & 0 & 1 & 1 & 0 \\
0 & 0 & 0 & 0 & 0 & 1 & 1 & 1 & 0 & 0 & 1 & 1 & 1 \\
0 & 0 & 0 & 0 & 0 & 0 & 0 & 1 & 1 & 1 & 0 & 1 & 1
\end{array}\right).$$

By {\sc Magma} computations, we can observe that $13$ is the shortest length for a binary minimal code of dimension $5$. 

For $q=3$ the code obtained is $[21,5,9]_3$, but with the aid of {\sc Magma} we found a $[20,5,9]_3$ minimal code. Hence, in general Construction \ref{pentagon} does not provide the smallest cutting blocking set in $PG(4,q)$.


In this sequel, we provide a construction of minimal codes of dimension 5 over $\F_2$, using cutting blocking sets in $PG(4,2)$, different from the pentagonal construction. We will refer to it as the \emph{hexagonal construction}.







\begin{construction}\label{esagono}
Let $\{P_1,P_2, P_3, P_4, P_5, P_6\}$ be  a projective frame in $PG(4,2)$. Without loss of generality, we can assume $P_1, P_2, P_3, P_4, P_5$ to be the (representatives of the) standard basis vectors and $P_6 = [1,1,1,1,1]$. Consider the lines $\ell_i = \langle P_i, P_{i+1} \rangle $ for $i\in\{1,2,3,4,5,6\}$, where the indices are taken modulo $6$. Let $Q:=[1,0,1,0,1]$.

The set $\mathcal{M}:=\ell_1\cup\ell_2\cup\ell_3\cup\ell_4\cup\ell_5\cup\ell_6\cup \{Q\}$ is a minimal cutting blocking set in $PG(4,2)$. This is not difficult to verify by hand or computer search.
\end{construction}

This construction produces the $[13,5,5]_2$ reduced minimal code generated by the following matrix:

$$ G_2=\left( \begin{array}{ccccccccccccc} 
1 & 0 & 0 & 0 & 0 & 1 & 1 & 0 & 0 & 0 & 1 & 0 & 1\\
0 & 1 & 0 & 0 & 0 & 1 & 1 & 1 & 0 & 0 & 1 & 1 & 0\\
0 & 0 & 1 & 0 & 0 & 1 & 0 & 1 & 1 & 0 & 1 & 1 & 1\\
0 & 0 & 0 & 1 & 0 & 1 & 0 & 0 & 1 & 1 & 1 & 1 & 0\\
0 & 0 & 0 & 0 & 1 & 1 & 0 & 0 & 0 & 1 & 0 & 1 & 1\\

\end{array} \right)$$

 With the aid of {\sc Magma} we  observed that the code constructed in this way and  the one obtained from the pentagonal construction are the only two $[13,5,5]_2$ minimal codes up to equivalence.

The hexagonal construction can be adapted to $q=3$. It gives a $[20,5,9]_3$ minimal codes, which is the shortest code that we could obtain.
Unfortunately, it seems difficult to generalize it for minimal codes of dimension $5$ over $\F_q$, for $q>3$.

\section{Conclusions and open problems}\label{sec:conclusion}

In this paper we characterized minimal linear codes from a geometrical point of view. Note that this characterization has been independently and simultaneously remarked also by Tang \emph{et al.} in \cite{tang2019characterization}.
This geometric approach allowed to prove new bounds on the length and the minimum distance of minimal codes, depending on their dimension and on the cardinality of the underlying field. 
Moreover, we proved that the family of minimal linear codes is asymptotically good. However, calculations in {\sc Magma} show
that our lower bound on the length is not sharp, so that there is still room for improvement.\\

\noindent \textbf{Problem 1.} Is it possible to prove a sharp lower bound on the length of a minimal linear code?\\

The already cited existence result of Chabanne \emph{et al.} of infinite sequences of minimal linear codes with fixed rate and growing length is unfortunately not constructive, and we are not aware of such a construction in the literature. The geometrical interpretation of minimal linear codes as cutting blocking sets should provide a way to construct codes with a length growing linearly in the dimension, by reducing as much as possible the number of points. However, in the construction by the tetrahedron of Theorem \ref{thm:tetraedro}, the length grows as the square of the dimension, and the arguments in Subsection \ref{subse:dim4} do not seem generalizable to higher dimensions.\\

\noindent \textbf{Problem 2.} Is it possible to give an explicit construction of an infinite sequence of minimal linear codes whose lengths are growing linearly in the dimension?\\

\noindent \textbf{Problem 3.} How to generalize Construction \ref{constrdim4} to dimension greater than $4$?\\

\noindent \textbf{Problem 4.} How to generalize the pentagonal construction (Construction \ref{pentagon}) to dimension greater than $5$?  How to generalize the hexagonal construction  (Construction \ref{esagono}) to every prime power $q$ and to dimension greater than $5$?

\medskip
Finally, in all our constructions of minimal codes, and in other constructions provided in \cite{bartoli2019minimalLin, bartoli2019inductive, bonini2019minimal, tang2019characterization}, we observed that the minimum distance satisfies $d\geq (k-1)(q-1)+1$, where this bound is met with equality in all our constructions of reduced minimal codes. Therefore, also motivated by Remark \ref{rem:minimumdistance}, where we observed that the bound of Theorem \ref{thm:dgeqkq2} is not sharp, we propose the following conjecture.

\begin{conjecture}
 Let $\C$ be an $[n,k,d]_q$ minimal code. Then $$d \geq (k-1)(q-1)+1.$$
\end{conjecture}

If the above conjecture is true, then, combining it with the Griesmer bound, as we did for Corollary \ref{coro:griesmer}, we would get a new lower bound on the length of minimal codes, namely 
$$n \geq (q-1)(k-1)+1 + \sum_{i=1}^{k-1} \left\lceil \frac{(q-1)(k-1)+1}{q^i}\right\rceil.$$
It is easy to see that this lower bound would improve Theorem \ref{prop:lowboundgeo} for every set of parameters. This is not in contrast with our experimental results, which show that the bound of Theorem \ref{prop:lowboundgeo} is not sharp.

\bibliographystyle{abbrv}
\bibliography{MinimalBiblio}

\end{document}